\newif\ifFull
\Fulltrue
\ifFull
\documentclass[12pt]{article}
\usepackage{fullpage}
\usepackage{amsthm}
\newtheorem{theorem}{Theorem}
\newtheorem{lemma}{Lemma}
\newtheorem{corollary}{Corollary}
\theoremstyle{definitions}

\usepackage{authblk}
\else
\documentclass[oribibl]{llncs}
\fi
\usepackage{amsmath, amssymb, algorithm, algorithmic, color}
\usepackage{enumerate, graphicx, hyperref}
\usepackage{cite}
\usepackage{wrapfig}
\graphicspath{{figures/}}

\ifFull
\else
\let\doendproof\endproof
\renewcommand\endproof{~\hfill\qed\doendproof}
\fi

\newcommand\NP{$\mathsf{NP}$}

\title{Parameterized Complexity of 1-Planarity}
\ifFull
\author[1]{Michael J. Bannister}
\author[2]{Sergio Cabello}
\author[1]{David Eppstein}
\affil[1]{Department of Computer Science, University of California, Irvine, USA}
\affil[2]{Faculty of Mathematics and Physics, University of Ljubljana, Slovenia}
\else
\author{Michael J. Bannister\inst{1}  \and Sergio Cabello\inst{2} \and David Eppstein\inst{1}}
\institute{Department of Computer Science, University of California, Irvine, USA \and
Faculty of Mathematics and Physics, University of Ljubljana, Slovenia}
\fi

\begin{document}
\maketitle

\begin{abstract}
We consider the problem of finding a $1$-planar drawing for a general graph, where a $1$-planar drawing is a drawing in which each edge participates in at most one crossing. Since this problem is known to be \NP-hard we investigate the parameterized complexity of the problem with respect to the vertex cover number, tree-depth, and cyclomatic number. For these parameters we construct fixed-parameter tractable algorithms. However, the problem remains \NP-complete for graphs of bounded bandwidth, pathwidth, or treewidth.
\end{abstract}

\pagestyle{plain}

\section{Introduction}
1-planar graphs (the graphs that can be drawn in the plane with at most one crossing per edge) were introduced by Ringel in 1965~\cite{Rin-AMSUH-65} and have since been extensively studied from the point of view of basic properties such as their colorings~\cite{Bor-MDA-84,CheKou-Alg-05}, edge density~\cite{Sch-MN-86,PacTot-Comb-97,BraEppGle-GD-12}, characterization by forbidden subgraphs~\cite{Kor-DM-08,KorMoh-JGT-13}, and embeddings on nonplanar surfaces~\cite{Suz-DM-10}. In graph drawing, 1-planarity has more recently become of interest, as a way of generalizing planar drawings in a controlled way that does not lead to too much visual complexity. Works in this area have compared 1-planarity to other forms of controlled crossings such as RAC (right-angle-crossing) graphs~\cite{EadLio-GD-11}, found an algorithmic characterization of the 1-planar drawings that can be straightened to have all edges represented by straight line segments~\cite{HonEadLio-COCOON-12}, and studied the transformation of rotation systems into 1-planar drawings~\cite{EadHonKat-GD-12}. However, until now there have been no published algorithms for finding 1-planar drawings of arbitrary graphs. Unfortunately, testing 1-planarity is \NP-hard in general~\cite{GriBod-Algo-07,KorMoh-JGT-13}, even for graphs obtained from planar graphs by adding a single edge~\cite{CabMoh-arxiv-12}, so we cannot expect it to be solved by an algorithm whose running time is a polynomial of the input size.

Because of the difficulty of recognizing 1-planar drawings, and their usefulness in graph drawing, it becomes of interest to study the complexity of algorithms for testing 1-planarity that are not fully polynomial. An important tool for this sort of study is \emph{parameterized complexity}~\cite{DowFel-99,FluGro-06}, according to which we seek additional numeric parameters (other than the numbers of edges and vertices) that measure the complexity of an input graph, and seek algorithms whose running time is the product of a polynomial in the input size and a non-polynomial function of the other parameter or parameters. If this can be accomplished, the result will in general be an algorithm that solves the problem correctly on all graphs, that can be relied on to be efficient for graphs that have small values of the parameter, and that has a performance that degrades gracefully as the parameter increases.

In this paper we study for the first time the parameterized complexity of 1-planarity. We provide a fixed-parameter tractable algorithm for the problem when it is parameterized by the cyclomatic number (the minimum number of edges that must be removed from the graph to make a forest) or the tree-depth. For a third parameter, the vertex cover number, we show that even more efficient FPT algorithms are possible, based on a polynomial kernel (a transformation of any instance to an equivalent instance with size polynomial in the parameter). However, as we show in an appendix, the problem remains \NP-complete for graphs of bounded bandwidth; therefore, it is unlikely that there exists a fixed-parameter tractable algorithm for 1-planarity when parameterized by bandwidth, pathwidth, treewidth, or clique-width. 

Although our primary motivation is in understanding the complexity of 1-planarity, our research on the vertex cover and tree-depth parameters has a secondary purpose as well, in exploring the circumstances in which general theorems that guarantee the existence of an inexplicit FPT algorithm (with unknown dependence on the parameter) can be made explicit. It is known that the graphs of bounded vertex cover number, and the graphs of bounded tree-depth, are well-quasi-ordered under induced subgraphs~\cite{NesOss-12}. This means that for any graph recognition problem closed under induced subgraphs (as 1-planarity is), and for any fixed bound on vertex cover or tree-depth, there is a finite set of forbidden induced subgraphs that can be used to characterize the problem, and a linear time recognition algorithm. However, the theorems that prove these results do not imply any computable bound on the size of these forbidden subgraphs or on the dependence on the parameter of these linear time algorithms. In contrast, for 1-planarity with these parameters we provide algorithms whose dependence on the parameter is known, explicit, and computable (albeit impractically large).

\section{Vertex cover number}

The \emph{vertex cover number} $k$ of an undirected graph $G$ is the minimum number of vertices needed to touch all of the edges of $G$. This number is central to the theory of parameterized complexity, to the point where Guo et al. call it ``the \emph{Drosophila} of fixed-parameter algorithmics''~\cite{GuoNieWer-WADS-05}. After much earlier work on the problem, the best fixed-parameter tractable algorithms for computing the vertex cover number, parameterized by this number, take time $O(1.2738^k+kn)$~\cite{CheKanXia-TCS-10}. We will show that, when parameterized by vertex cover number, 1-planarity is also fixed-parameter tractable, using a standard technique, kernelization, whereby we replace an instance graph by an equivalent instance of size bounded by a function of the kernel. Although the vertex cover number is a weaker parameter than the tree-depth that we consider later (a graph of vertex cover number $k$ has tree-depth at most $k+1$), we begin with this parameter for two reasons: (1) for this parameter we achieve stronger results, namely a polynomial kernel, than we do for the other parameters that we consider, and (2) the simplicity of this case makes it an appropriate warm-up for the other parameters.

\begin{lemma}[Czap and Hud\'ak~\cite{CzaHud-DAM-12}]
\label{lem:bipartite}
A complete bipartite graph is 1-planar if and only if it is of the form $K_{1,n}$, $K_{2,n}$, $K_{3,i}$ for $i\in\{3,4,5,6\}$, or $K_{4,4}$.
\end{lemma}

\begin{lemma}
\label{lem:exact-algorithm}
Testing 1-planarity of an $n$-vertex graph $G$ takes time $2^{O(n)}$.
\end{lemma}

\begin{proof}
If $G$ has more than $4n$ edges, we return that 
it is not 1-planar~\cite{PacTot-Comb-97}. 
Otherwise, we proceed with a divide and conquer algorithm:
the existence of cycle separators~\cite{Mil-JCSS-86} implies that in any 1-planar 
drawing there is a curve passing through $O(\sqrt{n})$ vertices that
separates the drawing into two balanced parts. We can then solve the problem
by solving $2^{O(n)}$ subproblems, each of them smaller
by a constant fraction. A detailed proof is in Appendix~\ref{app:lemma2}.
\end{proof}

\begin{lemma}
\label{lem:K2i-planar}
Let $G$ be a 1-planar graph, with a subgraph $H$ of the form $K_{2,i}$ formed by $i$ vertices of degree two, all with the same two neighbors. Then $G$ has a 1-planar drawing in which the induced drawing of $H$ is planar.
\end{lemma}

\begin{proof}
If two edges of $H$ that share an endpoint cross then we can uncross them, resulting in a drawing with fewer crossings, and if two non-incident edges of $H$ cross each other then we can redraw all of $H$ without crossings near the previous position of these two crossed edges, again reducing the total number of crossings. Therefore, a 1-planar drawing of $G$ that minimizes the total number of crossings has the desired property.
\end{proof}

\begin{figure}[t]
\ifFull
\centering\includegraphics[width=4in]{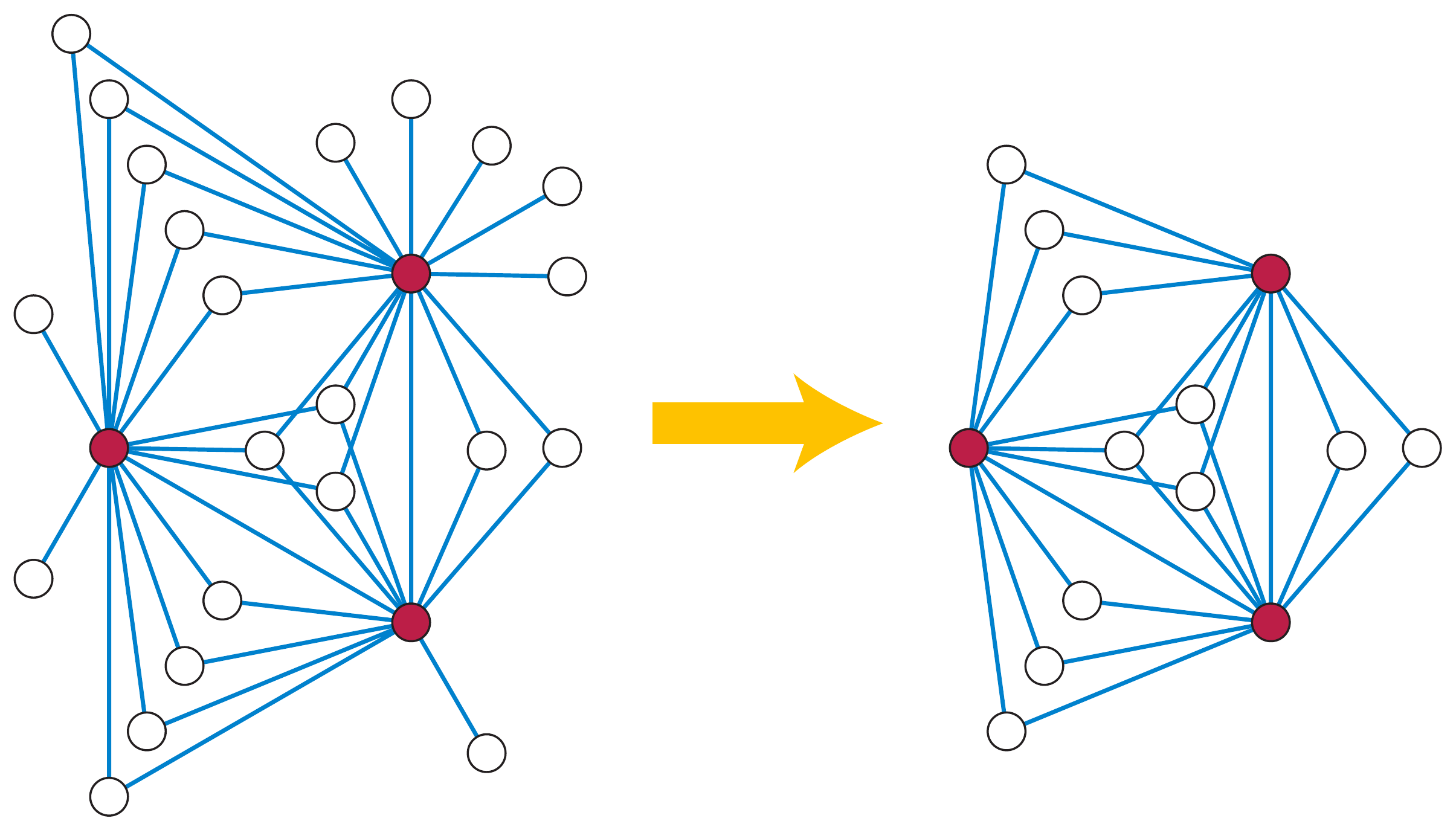}
\fi
\caption{Kernelization for vertex cover number $k$: remove degree-one vertices, and reduce each $K_{2,i}$  subgraph (with two cover vertices on one side of the bipartition) to $K_{2,\min\{i,2k-3\}}$. Here $k=3$, so the $K_{2,i}$ subgraphs are reduced to $K_{2,3}$.}
\label{fig:cover-kernel}
\end{figure}

\begin{lemma}
\label{lem:vc-kernel}
Let a graph $G$ have a known vertex cover $C$ of size $|C|=k$. Then in time $O(n)$ we can transform $G$ into an kernel $G_C$ of size $O(k^2)$ such that $G$ is 1-planar if and only if $G_C$ is 1-planar. A 1-planar drawing of $G_C$ may be transformed into a 1-planar drawing of $G$ in linear time.
\end{lemma}

\begin{proof}
\ifFull
If $G$ is to be 1-planar, there cannot be $5k$ two-edge paths connecting distinct pairs of vertices of $C$ through a different vertex of $G\setminus C$; otherwise smoothing out the internal vertices of those $5k$ paths we would obtain a drawing of a graph with $k$ vertices, $5k$ edges, and two crossings per edge, contradicting the bound of~\cite{PacTot-Comb-97}: If a graph on $n$ vertices has a a drawing with at most 2 crossings per edge, then it has at most $5n-10$ edges. This means that we can classify the vertices of $G\setminus C$ with degree three or more into $O(k)$ groups, according to two (arbitrary) neighbors in $C$. Moreover, assuming that $G$ is 1-planar, each of those groups has at most $6k$ vertices; otherwise some other common neighbor is repeated at least 7 times and thus $G$ contains a $K_{3,7}$. It follows that $G\setminus C$ has at most $O(k^2)$ vertices of degree three or more, if $G$ is 1-planar.
\fi
 
The vertices of degree
two in $G\setminus C$ can be grouped by radix sort according to the identities of their two neighbors in $C$, forming a collection of $K_{2,i}$ subgraphs. 
If $G$ is 1-planar, there are $O(k)$ such subgraphs $K_{2,i}$ because each of them
gives a two-edge path connecting two distinct vertices of $C$. 
If one of these $K_{2,i}$ subgraphs has $i > 2k-3$ then we claim that $G$ is 1-planar if and only if the subgraph $G'$ formed by deleting $i-(2k-3)$ vertices within this subgraph to form a smaller $K_{2,2k-3}$ subgraph is also 1-planar. In one direction, if $G$ is 1-planar, then clearly so is $G'$. In the other direction, suppose $G'$ is 1-planar; then by Lemma~\ref{lem:K2i-planar} it has a 1-planar drawing in which the given $K_{2,2k-3}$ subgraph is drawn planarly, with $2k-3$ quadrilateral faces. Two adjacent faces among this set of $2k-3$ must be empty of the $k-2$ vertices of $C$ that are not part of the $K_{2,2k-3}$ subgraph. Therefore, the two edges $e$ and $f$ separating these two faces cannot be crossed by any edge of the 1-planar drawing, for any crossing edge would either have to cross entirely across one of these two faces (violating 1-planarity) or have an endpoint in each of the two faces (violating the assumption that neither of these faces contains a vertex of $C$). The remaining vertices and edges of $G$ that were deleted to form the $K_{2,2k-3}$ subgraph may be added to the drawing, near path $ef$, without violating 1-planarity, showing as desired that $G$ is 1-planar.

Performing this replacement of $K_{2,i}$ by $K_{2,\min(i,2k-3)}$ separately for each of the groups of vertices in $G\setminus C$ results in the desired kernel $G_C$.  $G_C$ has $O(k^2)$ vertices of high degree and $O(k)$ groups of $O(k)$ vertices in $K_{2,i}$ subgraphs, for a total of $O(k^2)$ vertices.

If a drawing of $G_C$ is found, a corresponding drawing of $G$ may be found by eliminating crossings between pairs of edges belonging to the same $K_{2,i}$ subgraphs in $G_C$, finding an uncrossed length-two path with two vertices in $C$ as path endpoints within each $K_{2,2k-3}$ subgraph,  expanding each of these $K_{2,2k-3}$ subgraphs to $K_{2,i}$ for the correct value of $i$ from the original graph $G$ (placing the restored vertices near the uncrossed path), and finally adding back any deleted degree-one vertices of $G$.
\end{proof}

An example of this kernelization is depicted in Figure~\ref{fig:cover-kernel}, for a graph with vertex cover number three.

\begin{theorem}
\label{thm:cover-fpt}
We can test the 1-planarity of a given $n$-vertex graph, parameterized by its vertex cover number $k$, in time $O(n+2^{O(k^2)})$.
\end{theorem}

\begin{proof}
Apply an FPT algorithm to find an optimal vertex cover $C$, apply Lemma~\ref{lem:vc-kernel} to replace $G$ with a kernel $G_C$ of size $O(k^2)$, in linear time, and then apply Lemma~\ref{lem:exact-algorithm} to this kernel.

To reduce the dependence on $n$ in the time for the initial vertex cover step from $O(kn)$ to $O(n)$, we abort the algorithm if the input has more than $4n-8$ edges, and otherwise apply a standard kernelization for vertex cover: find a maximal matching $M$ in $G$, and
\ifFull
use $2|M|$ as a 2-approximation to the vertex cover. Find
\fi
all vertices of degree greater than $2|M|$; these must all belong to the optimal vertex cover, and can be removed from $G$, leaving a smaller graph $G'$ that has $O(k^2)$ edges (otherwise it could not be covered by the remaining low-degree vertices). Apply the vertex cover algorithm to $G'$ instead of to $G$.
\end{proof}

\begin{corollary}
We can test 1-planarity for split graphs in time $O(n)$.
\end{corollary}

\begin{proof}
If a given split graph has a clique of size seven, it is not 1-planar, and otherwise, it has a vertex cover of size six and we use the above algorithm.
\end{proof}

\section{Tree-depth}

As we now show, 1-planarity parameterized by tree-depth may be tested by an FPT algorithm.
The \emph{tree-depth} of a graph $G$ is the smallest depth of a forest $F$ on the same vertex set as $G$ such that every edge of $G$ connects an ancestor-descendant pair in $F$, where we measure the depth of a tree as the maximum number of vertices on a root-leaf path~\cite{NesOss-12}. Equivalently, it is the size of a maximum clique in a trivially perfect supergraph of $G$ chosen to minimize this clique size; here, a trivially perfect graph is the graph of ancestor-descendant pairs in a forest.
\ifFull
Since the trivially perfect graphs are a special case of the chordal graphs, and the treewidth of a graph is (one less than) the maximum size of a clique in a chordal supergraph chosen to minimize this clique size, it follows that tree-depth is always at least one plus treewidth.
\fi
A graph $G$ with vertex cover number $k$ has tree-depth at most $k+1$, for we may find a tree $T$ of depth $k+1$ that has the $k$ vertices of the cover on a path, from which all other vertices descend as leaves; all edges of $G$ connect ancestor-descendant pairs in $G$. For this reason, in some sense the result of this section is stronger than that of Theorem~\ref{thm:cover-fpt}, although the dependence on the parameter is worse.

An $n$-vertex path has tree-depth $\lceil\log_2(n+1)\rceil$. It follows that an arbitrary depth-first search tree for a given graph $G$ has a depth that is at most $2^d-1$ (because otherwise it would contain a path that is too long for the given depth) and at least the tree-depth $d$ (because the DFS tree has the ancestor-descendant property from which tree-depth is defined). Based on this observation, one can derive an FPT algorithm for computing the tree-depth, by finding a DFS tree, using it to construct a tree decomposition, and applying standard dynamic programming techniques to this decomposition~\cite{NesOss-12}.

\begin{lemma}
Let $G$ be a graph with tree-depth at most $d$, as witnessed by a forest $F$ of depth $d$ for which all edges of $G$ connect ancestor-descendant pairs. Then
in linear time it is possible to replace $G$ by an equivalent kernel for 1-planarity consisting of a collection of disconnected subgraphs with $O(2^{2d^2+O(d)})$ vertices each.
\end{lemma}

\begin{proof}
If $G$ is not biconnected we may test 1-planarity on each biconnected component of $G$ separately; 
therefore, we can assume without loss of generality that the
given graph $G$ is 2-connected, and that we have a tree $T$ of depth $d$
such that every edge of $G$ connects an ancestor-descendant pair in $T$.
We can also assume without loss of generality that each node of $T$ is
adjacent to at least one node in each of its child subtrees (because
otherwise we could move those children up to be siblings of the node,
which does not increase the depth) and that each child subtree induces
a connected subgraph (because otherwise we could split it into two
separate children). Because the tree-depth is $d$, the longest 
path in $G$ has length less than $2^d$.

Now consider how many children a node $v$ in $T$ can have. For each child
subtree $T_i$, consider the set $S_i$ of $v$ and ancestors of $v$ that are
connected to nodes in $T_i$. And for each subset $S$ of $v$ and its ancestors,
let $C(S)$ be the set of child subtrees $T_i$ of $v$ for which $S_i=S$. There
are at most $2^d$ different sets $S$, and we want to show that for each of
them, $C(S)$ has bounded size.
If $|S|=1$, this is easy: then $S=\{v\}$ (because otherwise there is no node
in $T_i$ adjacent to $v$) and $v$ is an articulation point, violating the
assumption of 2-connectivity.

Next, consider the case that $|S| \ge 3$. That is, we have a set $S$
consisting of v and two or more of its ancestors, and a set $C(S)$ of
child subtrees of $v$ that are each connected to all of the nodes in $S$.
Choose exactly three nodes of $S$ and, for each child subtree $T_i$ in
$C(S)$,  let $X_i$ be a smallest subgraph connecting the three chosen nodes
in the subgraph of $G$ induced by $T_i\cup S$. By the bound on the length of
paths in $G$, $|X_i|=O(2^d)$. Note that, among any three of
these trees $X_i$, $X_j$, and $X_k$ (all for members of $C(S)$) there must be
at least one crossing, because contracting each tree to a single node
produces a $K_{3,3}$ subgraph. There are $\Omega(|C(S)|^3)$ triples of trees,
at least one crossing per triple, and at most $|C(S)|$ triples that
involve any single crossing, so there are $\Omega(|C(S)|^2)$ crossings
altogether, among a set of only $O(|C(S)| 2^d)$ edges. In order to
prevent the pigeonhole principle from forcing some edge to be crossed
twice, we must have $|C(S)|=O(2^d)$.

Finally, consider the case that $|S|=2$. In this case, $|C(S)|$ can be
unbounded (e.g. consider the graph $K_{2,n}$, which has tree-depth
three). But, if it is greater than $2^d$, then it does not matter how
much greater it is: no cycle in the drawing can separate the two
vertices in $S$, because the minimal such cycle would have to have
length at most $2^d$ but would have to cross each of the subgraphs $T_i$,
a contradiction. So in this case we can split the graph into subgraphs
formed from each child $T_i$ together with an uncrossable edge between
the two nodes in $S$, and test 1-planarity separately for each of these
subgraphs. When $C(S)$ is small enough that no such split is possible,
$|C(S)|=O(2^d)$.

After performing any splits from the $|S|=2$ case, the remaining graph
has its nodes arranged into a tree of height $d$ in which each node has
$O(2^{2d})$ children. Therefore, the total number of nodes in the tree
is $O(2^{2d^2+O(d)})$.
\end{proof}

By combining this kernelization with the known FPT algorithm for computing tree-depth and with Lemma~\ref{lem:exact-algorithm} for testing the 1-planarity of the kernel, we obtain

\begin{theorem}
\label{thm:tree-depth}
The 1-planarity of a given graph, with tree-depth $d$, may be computed in time
$O(n2^{2^{2d^2+O(d)}})$.
\end{theorem}

\ifFull
Because the kernel may contain multiple connected components, and we are bounding the component size but not the number of components, the dependence of the time bound for $d$ is multiplied by $n$ rather than (as in Theorem~\ref{thm:cover-fpt}) added to it. Alternatively, it would be possible to remove isomorphic components from the kernel, and get a bound of the form $O(n+f(d))$, but with a larger dependence on~$d$.
\fi

As an example of the power of this approach, we show how to use it to recognize 1-planar cographs. Cographs are well-quasi-ordered by induced subgraphs~\cite{Dam-JGT-90}, from which it follows that there is an algorithm for testing 1-planarity by checking for the existence of a finite set of forbidden induced subgraphs; however, we do not know how to explicitly list these forbidden subgraphs nor do we know how to turn a recognition algorithm along such lines into an algorithm for finding a 1-planar drawing. In contrast, the algorithm outlined below for recognizing 1-planar cographs is explicit (albeit with impractically large constants) and constructs a drawing of the graph.

\begin{figure}[t]
\ifFull
\centering\includegraphics[height=1.75in]{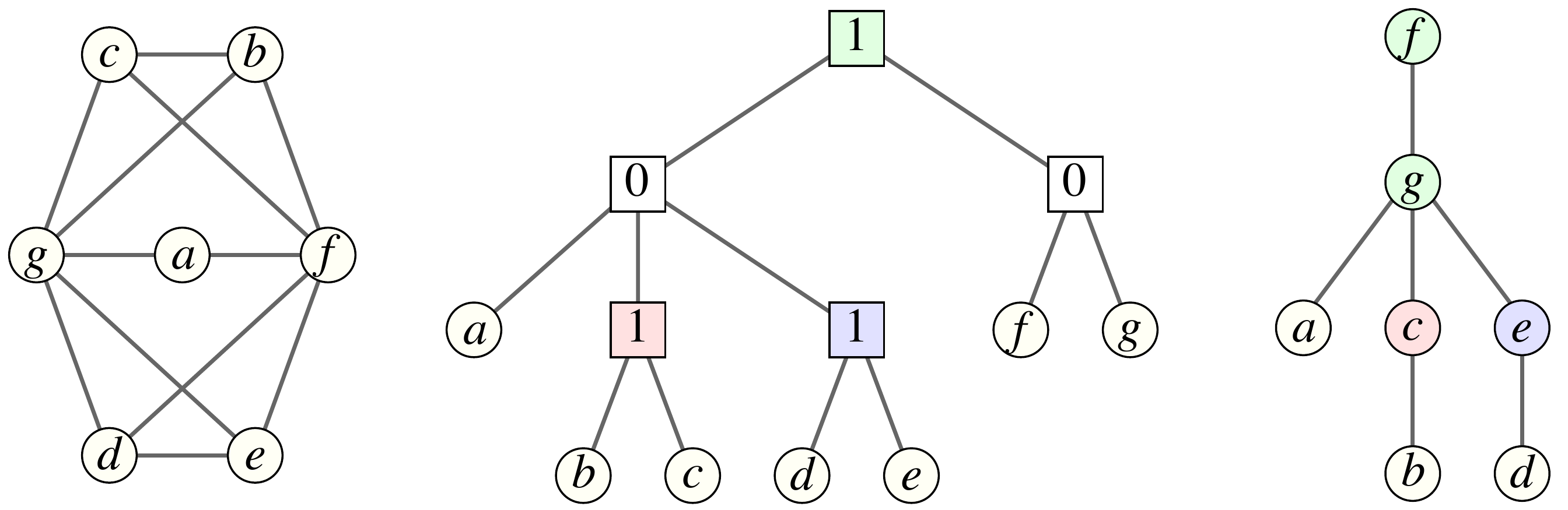}
\fi
\caption{Finding a low-tree-depth representation of a cograph by forming a path for each 1-labeled cotree node, consisting of the cotree leaves that descend from it but are not in its heaviest child. Left: a cograph. Center: its cotree. Right: the tree formed by connecting together the paths $L_x$. Each cotree node has the same color as its corresponding path.}
\label{fig:cograph-treedecomp}
\end{figure}

\begin{lemma}
\label{lem:cograph}
Let $\mathcal{C}_{a,b}$ denote the class of cographs that do not contain $K_a$ nor $K_{b,b}$ as subgraphs. Then the graphs in $\mathcal{C}_{a,b}$ have tree-depth at most $1+(a-1)(b-1)$.
\end{lemma}

\begin{proof}
For any graph $G$ in this class, we use a cotree representing $G$, and use it to guide the construction of a forest $F$ on the nodes of $G$. A cotree has the vertices of $G$ as its leaves; every internal node is labeled either $0$ or $1$, and two vertices of $G$ are adjacent if and only if their lowest common ancestor is labeled $1$. We assume that this tree is in canonical form meaning that no two adjacent internal nodes have the same label as each other and that each internal node has at least two children.

For each node $x$ labeled 1 in the cotree, let $H_x$ denote the subtree descending from a child of $x$ that contains the largest number of leaves (breaking ties arbitrarily) and let $L_x$ denote the set of leaf descendants of $x$ that are not in $H_x$. For each maximal set $L_x$ (not contained in $L_y$ for some 1-labeled node $y$), we form a path, which will form a subgraph of $F$. If the closest 1-labeled ancestor of cotree node $x$ is node $y$, we set the parent of the top node of path $L_x$ to be the bottom node of path $L_y$. In addition, if any vertex $v$ of $G$ does not belong to a set $L_x$, we make it a leaf of the forest $F$, and we set the parent of $v$ to be the bottom node of the path for the lowest 1-labeled ancestor of $v$ in the cotree. The forest constructed in this way (shown in Figure~\ref{fig:cograph-treedecomp}) will necessarily have the defining property of tree-depth that every edge in $G$ connects an ancestor-descendant pair in $F$.

If $L_x$ has at least $2b-1$ leaves, then the leaf descendants of $x$ contain a $K_{b,b}$ subgraph. For this reason, every path $L_x$ has at most $2(b-1)$ vertices of $G$ in it. Additionally, on any path from the root to a leaf in the cotree, at most one of the 1-labeled cotree nodes can have more than $b-1$ nodes in $L_x$, for if one such node does, then each of its ancestors must have at most $b-1$ nodes in $L_x$, or else we would again have a $K_{b,b}$ subtree. Finally, observe that a path from the root to a leaf in the cotree that has $a-1$ 1-labeled nodes would give rise to a $K_a$ subgraph; therefore, every such path has at most $a-2$ 1-labeled nodes.
By this analysis, the longest path from leaf to root that could exist in the forest $F$ consists of one vertex of $G$ that does not belong to a set $L_x$, one set $L_x$ of size $2(b-1)$, and $a-3$ sets $L_x$ of size $b-1$, matching the depth given in the statement of the lemma.
\end{proof}

\begin{corollary}
We can recognize 1-planar cographs, and find 1-planar drawings of them, in $O(n)$ time.
\end{corollary}

\begin{proof}
We first test whether the given cograph contains $K_7$ or $K_{5,5}$ as a subgraph. If it does, it is not 1-planar. If it does not, we may apply Lemma~\ref{lem:cograph} and Theorem~\ref{thm:tree-depth}.
\end{proof}

\section{Cyclomatic number}

We say that a graph $G$ has \emph{cyclomatic number} $k$ if $k$ is the smallest number of edges that must be removed from $G$ to yield a forest; equivalently $k=m - n + c$, where $c$ is the number of connected components in $G$. By a \emph{maximal degree two path} we shall mean a path between two vertices each of degree greater than two such that all vertices in the interior of the path have degree two. For technical reasons, an edge between vertices each having degree greater than two will also be considered a maximal degree two path.
Gurevich et al{.} define a \emph{$k$-almost-tree} to be a graph $G$ such that given a spanning tree $T$ of $G$ every biconnected component of $G$ has at most $k$ edges not in $T$\cite{GurStoVis-JACM-84}.
This is equivalent to each biconnected component having cyclomatic number $k$.

The cyclomatic number and $k$-almost-tree parameter have previously been used as parameters in fixed parameter algorithms. For example, in biology,  gene expression can be represented as a Boolean network in which individual genes are represented as vertices and edges represent correlations between pairs of genes.  Fixed parameter tractable algorithms have been designed for the control problem, which involves finding sequences of valid labelings of genes as being active or inactive~\cite{AkuHayChi-JTB-07}. In operations research, fixed-parameter algorithms for the continuous facility location problem have been constructed, where weighted edges represent a road network on which to efficiently place facilities serving clients in the network~\cite{GurStoVis-JACM-84}.  Intraprogram communication networks in distributed systems use vertices to represent modules of a program to be computed in parallel and edges to represent communicating pairs of modules; they also have structure yielding fixed-parameter algorithms~\cite{Fer-TSE-89} with respect to this parameter.

\begin{lemma}\label{lem:vertex-count}
If $G$ is a graph with cyclomatic number $k$ and no degree one vertices, then $G$ has at most $2k - 2$ vertices of degree greater than two. Furthermore, this bound is tight. Also, the number of maximal degree two paths is at most $3k-3$.
\end{lemma}
\begin{proof}
Double counting edges yields $2(n - c + k) \geq 2a + 3b$, where $a$ is number of degree two vertices and $b$ is the number of vertices of degree greater than two. Using $n=a+b$ and $c\ge 1$ we obtain $b \leq 2k - 2$, establishing the upper bound. For the upper bound consider any biconnencted cubic graph with $2k-2$ vertices, e.g., a cubic Halin graph whose characteristic tree has $k$ leaves.

For the bound on the maximal degree two paths consider the graph $G'$ where each maximal degree two path is reduced to a single edge. The graph $G'$ has cyclomatic number $k$ and at most $2k - 2$ vertices. This implies that $G'$ has at most $3k - 3$ edges, establishing the bound.
\end{proof}

\begin{lemma}\label{lem:uncrossing}
If $G$ is a 1-planar, then there is a 1-planar drawing of $G$ such that maximal degree two paths do not self intersect.
\end{lemma}
\begin{proof}
It suffices to show that a self crossing in a maximal degree two path can be removed without increasing the number of crossings on any edges.  We can locally uncross a self intersection changing the drawing within  a circular region $\mathcal{R}$ around the intersection that is not crossed by other edges.  See Figure~\ref{fig:almost-tree-path-crossing} for an example of this operation. 
\end{proof}
\begin{figure}[t]
\ifFull
\centering\includegraphics[width=0.8\textwidth]{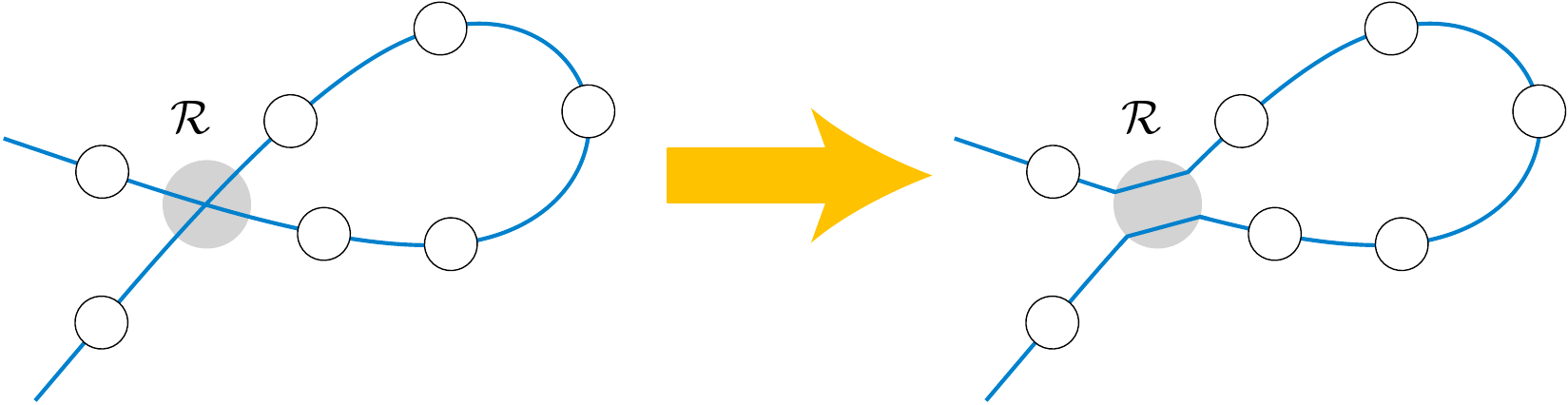}
\fi
\caption{Removing a crossing in a degree two path}
\label{fig:almost-tree-path-crossing}
\end{figure}

\begin{lemma}\label{lem:sequence-reduction}
Every word on $n>1$ symbols, without consecutive equal symbols, of length greater than $2n! - 1$ has a subword on $k>1$ symbols, for some $k\le n$, such that each symbol appears at least $k$ times in the subword. Furthermore, this bound is tight, i.e., there exists a word $w$ of length $2n!-1$ on $n$ symbols such that for every $1<k\le n$, $w$ has no subword on $k$ symbols in which each symbol appears at least $k$ times.
\end{lemma}
\begin{proof}
Let $w$ be a word on $n$ symbols of length at least $2(n!) - 1$, and let $\sigma$ be the symbol appearing least often in $w$. If $\sigma$ occurs more than $n$ times in $w$, then we are done. So assume that $\sigma$ occurs at most $n - 1$ times. Removing $\sigma$ from $w$ leaves us with at least $2(n!) - n$ symbols split into at most $n$ subwords. Thus, the longest of these subwords, call it $u$, has length at least $(2(n!) - n) / n = 2(n-1)!  - 1$. Since $u$ contains at most $n-1$ unique symbols we are done by induction on $n$.

To construct a word on $n$ symbols of length $2(n!) - 1$ with no reducible subword, let $\sigma_0, \sigma_1, \ldots, \sigma_n$ be our $n$ symbols. Now recursively define the words by $w_k = (w_{k-1}\sigma_k)^{k-1}w_{k-1}$ and $w_2 = \sigma_0\sigma_1\sigma_0$. A simple induction argument shows that the length of $w_k$ is $2(k!) - 1$.
\end{proof}

\begin{figure}[t]
\centering\includegraphics[width=0.95\textwidth]{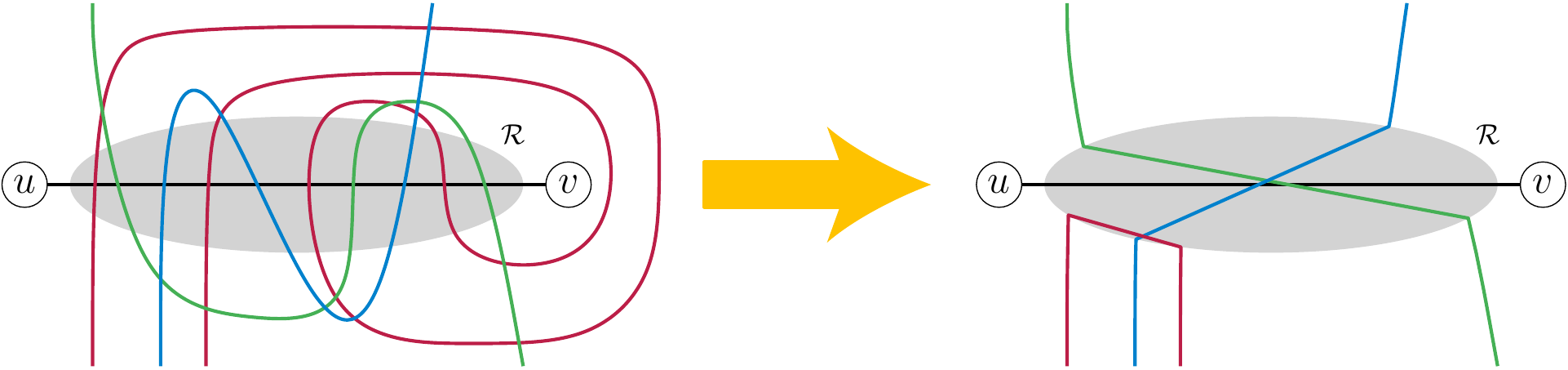}
\caption{Left crossing sequence \texttt{rgbrbrgbrg}; Right crossing sequence \texttt{bg}}
\label{fig:almost-tree-crossing-sequence}
\end{figure}

\begin{lemma}\label{lem:crossing-bound}
If $G$ is a 1-planar graph with $p$ maximal degree two paths, then $G$ has a 1-planar drawing such that every maximal degree two path is crossed at most $2p! - 1$ times.
\end{lemma}
\begin{proof}
We need only show that given a maximal degree two path from $u$ to $v$ with more than $2p! - 1$ crossings, we can reduce the number of times that it is crossed without increasing the crossing count on other degree two paths.

First, we continuously deform the plane such that the path from $u$ to $v$ is a straight line. This is possible since we may assume that maximal degree two paths do not self intersect by Lemma~\ref{lem:uncrossing}. Now we consider the sequence of crossings through the path from $u$ to $v$ my other maximal degree two paths. In this sequence there are at most $p$ symbols. So if the number of crossings on the path from $u$ to $v$ is greater than $2p! - 1$, Lemma~\ref{lem:sequence-reduction} implies that there is a subword on $p'$ symbols such that every symbol appears at least $p'$ times. 

Now, we construct a strictly convex region $\mathcal{R}$ around the crossings represented by this word such that only paths represented in the word intersect the region, and such that a path does not reintersect the path from $u$ to $v$ without first leaving $\mathcal{R}$. For every path we shortcut it from the first time it intersects $\mathcal{R}$ to the last time it intersects $\mathcal{R}$, in path order, with a straight line. So now each path in $\mathcal{R}$ is a straight line, and therefore they can only intersect each other at most once. So, we have reduced the number of crossings on the path from $u$ to $v$, without increasing the crossings on the other paths.
\end{proof}

\begin{lemma}
Let $G$ be a graph with cyclomatic number $k$. Then in linear time we can transform $G$ into a kernel $G_C$ of size $O((3k-3)(3k-3)!)$ such that $G$ is $1$-planar if and only if $G_C$ is $1$-planar. In addition, a $1$-planar drawing of $G_C$ may be transformed into a $1$-planar drawing of $G$ in linear time.
\end{lemma}
\begin{proof}
We remove degree one vertices from $G$ until no more are left, producing the \emph{$2$-core} of $G$~\cite{Sei-SN-83}. This process can be done in linear time by maintaining a queue of degree one vertices. A degree one vertex may be added to any drawing without introducing crossings, so a graph has a $1$-planar drawing if and only if its $2$-core has a $1$-planar drawing.

Lemma~\ref{lem:vertex-count} implies that we have  at most $p = 3k - 3$ maximal degree two paths. For each of these maximal degree two paths we reduce the number of degree two vertices to $2p! + 1$ if they exceed this amount. Since Lemma~\ref{lem:crossing-bound} guarantees that, if $G$ is $1$-planar, then it has a drawing such that no maximal degree two path is crossed more than $2p! - 1$ times, this reduction does not change the $1$-planarity of the graph. Thus, we have a kernel $G_C$ of size $O((3k-3)(3k-3)!)$ such that $G$ is $1$-planar if and only if $G_C$ is $1$-planar.
\end{proof}

\begin{theorem}\label{thm:cyclomatic}
We can test the $1$-planarity of a graph with cyclomatic number $k$ in time $O\left(n + 2^{O((3k)!)}\right)$.
\end{theorem}

Since a graph can be decomposed into its biconnected components in linear time and edges in separate biconnected components need not cross we have the following corollary to Theorem~\ref{thm:cyclomatic}.

\begin{corollary}
We can test the $1$-planarity of a $k$-almost tree in time $O(n2^{O((3k)!)})$.
\end{corollary}

\subsection*{Acknowledgements}
The research of Bannister and Eppstein was supported in part  by the National Science Foundation under grants 0830403 and 1217322, and by the Office of Naval Research under MURI grant N00014-08-1-1015. 
The research of Cabello was supported in part by the Slovenian Research Agency, program P1-0297, project
J1-4106, and within the EUROCORES Programme EUROGIGA (project GReGAS) of the European Science Foundation. We also gratefully acknowledge the Slovenian Research Agency for travel funds allowing the authors to meet and perform this research.

{\raggedright
\bibliographystyle{abuser}
\bibliography{one-planarity}}

\begin{thebibliography}{10}

\bibitem{AkuHayChi-JTB-07}
T.~Akutsu, M.~Hayashida, W.-K. Ching, and M.~K. Ng.
\newblock {Control of Boolean networks: Hardness results and algorithms for
  tree structured networks}.
\newblock {\em J. Theor. Biol.} 244(4):670{--}679, 2007,
  \href{http://dx.doi.org/10.1016/j.jtbi.2006.09.023}%
{doi:10.1016/j.jtbi.2006.09.023}.

\bibitem{Bor-MDA-84}
O.~V. Borodin.
\newblock {Solution of the Ringel problem on vertex-face coloring of planar
  graphs and coloring of 1-planar graphs}.
\newblock {\em Metody Diskret. Analiz.} no.~41, pp.~12{--}26, 108, 1984.

\bibitem{BraEppGle-GD-12}
F.~J. Brandenburg, D.~Eppstein, A.~Glei{\ss}ner, M.~T. Goodrich, K.~Hanauer,
  and J.~Reislhuber.
\newblock {On the density of maximal 1-planar graphs}.
\newblock {\em Proc. 20th Int. Symp. Graph Drawing}, 2013.

\bibitem{CabMoh-arxiv-12}
S.~Cabello and B.~Mohar.
\newblock {Adding one edge to planar graphs makes crossing number and
  1-planarity hard}.
\newblock {\em CoRR} abs/1203.5944, 2012.

\bibitem{CheKanXia-TCS-10}
J.~Chen, I.~A. Kanj, and G.~Xia.
\newblock {Improved upper bounds for vertex cover}.
\newblock {\em Theoretical Computer Science} 411(40-42):3736{--}3756, 2010,
  \href{http://dx.doi.org/10.1016/j.tcs.2010.06.026}%
{doi:10.1016/j.tcs.2010.06.026}.

\bibitem{CheKou-Alg-05}
Z.-Z. Chen and M.~Kouno.
\newblock {A linear-time algorithm for 7-coloring 1-plane graphs}.
\newblock {\em Algorithmica} 43(3):147{--}177, 2005,
  \href{http://dx.doi.org/10.1007/s00453-004-1134-x}%
{doi:10.1007/s00453-004-1134-x}.

\bibitem{CzaHud-DAM-12}
J.~Czap and D.~Hud{\'a}k.
\newblock {1-planarity of complete multipartite graphs}.
\newblock {\em Discrete Applied Mathematics} 160(4-5):505{--}512, 2012,
  \href{http://dx.doi.org/10.1016/j.dam.2011.11.014}%
{doi:10.1016/j.dam.2011.11.014}.

\bibitem{Dam-JGT-90}
P.~Damaschke.
\newblock {Induced subgraphs and well-quasi-ordering}.
\newblock {\em J. Graph Th.} 14(4):427{--}435, 1990,
  \href{http://dx.doi.org/10.1002/jgt.3190140406}%
{doi:10.1002/jgt.3190140406}.

\bibitem{DowFel-99}
R.~G. Downey and M.~R. Fellows.
\newblock {\em {Parameterized Complexity}}.
\newblock Monographs in Computer Science. Springer, 1999,
  \href{http://dx.doi.org/10.1007/978-1-4612-0515-9}%
{doi:10.1007/978-1-4612-0515-9}.

\bibitem{EadHonKat-GD-12}
P.~Eades, S.-H. Hong, N.~Katoh, G.~Liotta, P.~Schweitzer, and Y.~Suzuki.
\newblock {Testing maximal 1-planarity of graphs with a rotation system in
  linear time}.
\newblock {\em Proc. 20th Int. Symp. Graph Drawing}, 2013.

\bibitem{EadLio-GD-11}
P.~Eades and G.~Liotta.
\newblock {Right angle crossing graphs and 1-planarity}.
\newblock {\em Graph Drawing: 19th International Symposium, GD 2011, Eindhoven,
  The Netherlands, September 21-23, 2011, Revised Selected Papers},
  pp.~148{--}153. Springer, Lecture Notes in Computer Science 7034, 2012,
  \href{http://dx.doi.org/10.1007/978-3-642-25878-7\_15}%
{doi:10.1007/978-3-642-25878-7\_15}.

\bibitem{Fer-TSE-89}
D.~Fernandez-Baca.
\newblock {Allocating modules to processors in a distributed system}.
\newblock {\em Software Engineering, IEEE Transactions on} 15(11):1427{--}1436,
  1989, \href{http://dx.doi.org/10.1109/32.41334}%
{doi:10.1109/32.41334}.

\bibitem{FluGro-06}
J.~Flum and M.~Grohe.
\newblock {\em {Parameterized Complexity Theory}}.
\newblock Texts in Theoretical Computer Science. Springer, 2006.

\bibitem{GriBod-Algo-07}
A.~Grigoriev and H.~L. Bodlaender.
\newblock {Algorithms for graphs embeddable with few crossings per edge}.
\newblock {\em Algorithmica} 49(1):1{--}11, 2007,
  \href{http://dx.doi.org/10.1007/s00453-007-0010-x}%
{doi:10.1007/s00453-007-0010-x}.

\bibitem{GuoNieWer-WADS-05}
J.~Guo, R.~Niedermeier, and S.~Wernicke.
\newblock {Parameterized complexity of generalized vertex cover problems}.
\newblock {\em 9th International Workshop, WADS 2005, Waterloo, Canada, August
  15-17, 2005, Proceedings}, pp.~36{--}48. Springer, Lecture Notes in Computer
  Science 3608, 2005, \href{http://dx.doi.org/10.1007/11534273\_5}%
{doi:10.1007/11534273\_5}.

\bibitem{GurStoVis-JACM-84}
Y.~Gurevich, L.~Stockmeyer, and U.~Vishkin.
\newblock {Solving NP-Hard Problems on Graphs That Are Almost Trees and an
  Application to Facility Location Problems}.
\newblock {\em J. ACM} 31(3):459{--}473, June 1984,
  \href{http://dx.doi.org/10.1145/828.322439}%
{doi:10.1145/828.322439}.

\bibitem{HonEadLio-COCOON-12}
S.-H. Hong, P.~Eades, G.~Liotta, and S.-H. Poon.
\newblock {F{\'a}ry's theorem for 1-planar graphs}.
\newblock {\em Computing and Combinatorics: 18th Annual International
  Conference, COCOON 2012, Sydney, Australia, August 20-22, 2012, Proceedings},
  pp.~335{--}346. Springer, Lecture Notes in Computer Science 7434, 2012,
  \href{http://dx.doi.org/10.1007/978-3-642-32241-9\_29}%
{doi:10.1007/978-3-642-32241-9\_29}.

\bibitem{KapSha-SJC-96}
H.~Kaplan and R.~Shamir.
\newblock {Pathwidth, bandwidth, and completion problems to proper interval
  graphs with small cliques}.
\newblock {\em SIAM J. on Computing} 25(3):540{--}561, 1996,
  \href{http://dx.doi.org/10.1137/S0097539793258143}%
{doi:10.1137/S0097539793258143}.

\bibitem{Kor-DM-08}
V.~P. Korzhik.
\newblock {Minimal non-1-planar graphs}.
\newblock {\em Discrete Mathematics} 308(7):1319{--}1327, 2008,
  \href{http://dx.doi.org/10.1016/j.disc.2007.04.009}%
{doi:10.1016/j.disc.2007.04.009}.

\bibitem{KorMoh-JGT-13}
V.~P. Korzhik and B.~Mohar.
\newblock {Minimal Obstructions for 1-Immersions and Hardness of 1-Planarity
  Testing}.
\newblock {\em J. Graph Th.} 72(1):30--71, 2013,
  \href{http://dx.doi.org/10.1002/jgt.21630}%
{doi:10.1002/jgt.21630}.

\bibitem{Mil-JCSS-86}
G.~L. Miller.
\newblock {Finding Small Simple Cycle Separators for 2-Connected Planar
  Graphs}.
\newblock {\em J. Comput. Syst. Sci.} 32(3):265--279, 1986,
  \href{http://dx.doi.org/10.1016/0022-0000(86)90030-9}%
{doi:10.1016/0022-0000(86)90030-9}.

\bibitem{NesOss-12}
J.~Ne{\v{s}}et{\v{r}}il and P.~Ossona~de Mendez.
\newblock {\em {Sparsity: Graphs, Structures, and Algorithms}}.
\newblock Algorithms and Combinatorics~28. Springer, 2012, pp.~115{--}144,
  \href{http://dx.doi.org/10.1007/978-3-642-27875-4}%
{doi:10.1007/978-3-642-27875-4}.

\bibitem{PacTot-Comb-97}
J.~Pach and G.~T{\'o}th.
\newblock {Graphs drawn with few crossings per edge}.
\newblock {\em Combinatorica} 17(3):427{--}439, 1997,
  \href{http://dx.doi.org/10.1007/BF01215922}%
{doi:10.1007/BF01215922}.

\bibitem{Rin-AMSUH-65}
G.~Ringel.
\newblock {Ein Sechsfarbenproblem auf der Kugel}.
\newblock {\em Abhandlungen aus dem Mathematischen Seminar der Universit{\"a}t
  Hamburg} 29:107{--}117, 1965, \href{http://dx.doi.org/10.1007/BF02996313}%
{doi:10.1007/BF02996313}.

\bibitem{Sch-MN-86}
H.~Schumacher.
\newblock {Zur Struktur 1-planarer Graphen}.
\newblock {\em Mathematische Nachrichten} 125:291{--}300, 1986.

\bibitem{Sei-SN-83}
S.~B. Seidman.
\newblock {Network structure and minimum degree}.
\newblock {\em Social Networks} 5(3):269{--}287, 1983,
  \href{http://dx.doi.org/10.1016/0378-8733(83)90028-X}%
{doi:10.1016/0378-8733(83)90028-X}.

\bibitem{Suz-DM-10}
Y.~Suzuki.
\newblock {Optimal 1-planar graphs which triangulate other surfaces}.
\newblock {\em Discrete Mathematics} 310(1):6{--}11, 2010,
  \href{http://dx.doi.org/10.1016/j.disc.2009.07.016}%
{doi:10.1016/j.disc.2009.07.016}.

\end{thebibliography}

\ifFull
\newpage
\appendix

\section{Proof of Lemma~\ref{lem:exact-algorithm}.}
\label{app:lemma2}

\begin{proof}[Lemma~\ref{lem:exact-algorithm}]
If the graph has more than $4n-8$ edges, we immediately return that 
it is not 1-planar~\cite{PacTot-Comb-97}. 
Otherwise, we proceed with a divide and conquer algorithm, as follows.

Consider a drawing of $G$ and a circular sequence $\pi$ of vertices of $G$.
A \emph{$\pi$-curve} (for the drawing) is a simple, closed curve in the plane visiting
the vertices of $\pi$, in the order given by $\pi$, and otherwise disjoint from the drawing.
In particular, the curve does not intersect the interior of any edge.

We first use cycle separators to argue the following: 
in any 1-planar drawing of $G$ there exists a $\pi$-curve, with $\pi$ of length $O(\sqrt{n})$,
that has at most $2|E|/3$ edges in the interior and at most $2|E|/3$ edges in the exterior. 
We call such curve a \emph{balanced separating curve} for the drawing.
The existence of such curve follows from the result of Miller~\cite{Mil-JCSS-86}. Consider the 
planarization $G_P$ of the 1-planar drawing of $G$, where each intersection is replaced by a vertex,
take its vertex-face incidence embedded graph $\Gamma$, which is 2-connected, and consider the 
curve $\alpha$ described by a cycle separator in $\Gamma$ splitting the faces of $\Gamma$ 
in a balanced way. Since each face of $\Gamma$ corresponds to an edge of $G_P$, the
curve $\alpha$ gives a balanced partition of the edges of $G_P$. 
The curve $\alpha$ passes through vertices of $G_P$, so it may pass through some crossing of the 
drawing. For each crossing, say between edges $e$ and $e'$, 
we make a local rerouting of $\alpha$ in the neighborhood of $e$ and $e'$ so
that it intersects the drawing of $G$ only at the endpoints of $e$ and $e'$. 
Such rerouting is possible because 1-planarity implies that $e$ and $e'$ 
cannot participate in any other crossing. See Figure~\ref{fig:separator}.
Each such rerouting replaces a portion of the curve passing through a crossing by a portion passing
through at most two vertices of $G$, so the resulting curve passes through $O(\sqrt{n})$ vertices
of $G$. This finishes the proof of existence of balanced separating curves in 1-planar drawings.

\begin{figure}[t]
\centering\includegraphics[width=.7\textwidth]{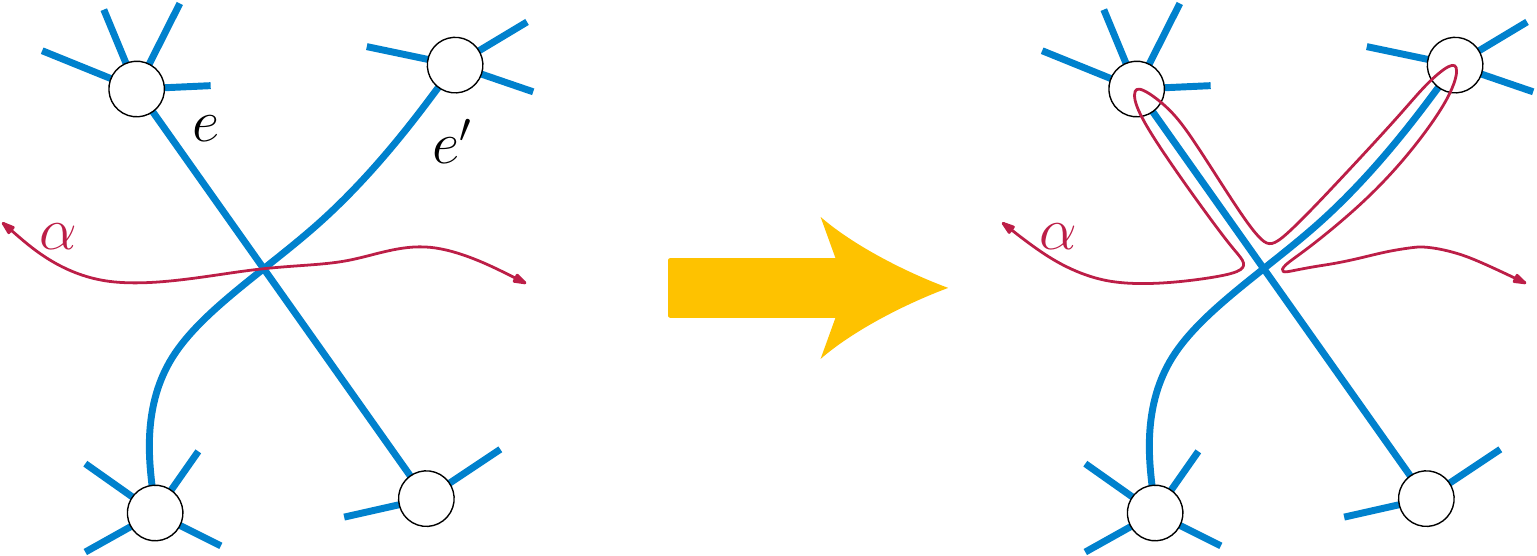}
\caption{Converting a separator for the planarization $G_P$ into a balanced separating curve.}
\label{fig:separator}
\end{figure}

For the divide and conquer algorithm, 
let us consider a more general problem: given an $n$-vertex graph $G=(V,E)$ and a subset $F\subseteq E$,
is there a 1-planar drawing of $G$ where no edge of $F$ participates in any crossing? Let us tell that $\phi(G,F)$ is true when such drawing exists, and false otherwise. We can compute $\phi(G,F)$ recursively by trying all possible balanced separating curves and edge partitions, as follows.

Consider a circular sequence $\pi=u_0,\dots ,u_k,u_0$ of vertices of $G$
and a partition of the edge set $E=E_1\sqcup E_2$.
Let $H$ be the graph obtained from the cycle $u_0\dots u_ku_0$ by adding a vertex $u'$ connected to all vertices $u_0,\dots, u_k$.
For $i=1,2$, let $G_i=H+E_i$ and $F_i=E(H)\cup (F\cap E_i)$.
It holds $\phi(G_1,F_1)$ and $\phi(G_2,F_2)$ are true if and only if $G$ has a 1-planar drawing where no edge of $F$ participates in a crossing and there is $\pi$-curve separating $E_1$ from $E_2$. Moreover, in linear time we can combine drawings of $G_1$ and $G_2$ certifying that $\phi(G_1,F_1)$ and $\phi(G_2,F_2)$ are true to obtain a 1-planar drawing of $G$ certifying that $\phi(G,F)$ is true. 

Because of the existence of balanced separating curves for $1$-planar drawings we have  
\[
	\phi(G,F) ~=~ \bigvee_{\pi, E_1,E_2} \left(  \phi(G_1,F_1)\wedge \phi(G_2,E_2)\right),
\]
where $\pi$ ranges over all sequences of $O(\sqrt{n})$ distinct vertices and $E_1,E_2$
over all partitions of $E$ with at most $2|E|/3$ edges each.
This means that $\phi(G,F)$ can be obtained solving $O(n^{\sqrt{n}}2^{|E|})=2^{O(n)}$ subproblems,
each with at most $2|E|/3+O(\sqrt{n})$ edges. We thus get, when $|E|$ is larger than some constant,
the recursion
$T(|E|)\le 2^{O(n)} T(2|E|/3+O(\sqrt{n}))$, which solves to $T(|E|)\le 2^{O(n)}$.
\end{proof}

\section{Bandwidth}

If the vertices of a graph $G$ are arranged on the real line with distinct integer coordinates, the \emph{bandwidth} of the arrangement is the maximum length of an edge of $G$. The bandwidth of the graph $G$ itself is the minimum, over all possible linear arrangements of $G$, of the length of the longest edge in the arrangement. The bandwidth may also be defined as one less than the minimum clique number of any proper interval graph having $G$ as a subgraph, a formulation that makes clear the relation between bandwidth, pathwidth (the same notion with interval graphs in place of proper interval graphs), and treewidth (the same notion with chordal graphs in place of interval graphs)~\cite{KapSha-SJC-96}.

In this section we show that 1-planarity remains \NP-complete even when restricted to graphs of bounded bandwidth.
Graphs of bounded bandwidth also have bounded pathwidth, treewidth, and clique-width, so 1-planarity is also hard for those parameters.

\subsection{Overview}
\label{sec:npc-ov}

Our proof of \NP-completeness of 1-planarity for graph planarity is based on a standard gadget-based reduction from 3-satisfiability, but because of the complexity of the gadgets we break the proof up into several steps. In this subsection we outline the general idea of the reduction.

The overall structure is a graph $G$ of bounded bandwidth with three parts: one part for the clauses of the 3-satisfiability instance (blue in Figure~\ref{fig:npc-overview}), one part for the variables of the instance (red in the figure), and one part that (despite its low bandwidth) forms a grid-like structure that holds the other two parts in their places (black in the figure). The variable and clause gadgets will form subgraphs that are only attached to the grid gadget at one end (with the points of attachment not all lying within a single face of the grid); the points of attachment are shown as small green circles in the upper left of the figure.

Although not attached graph-theoretically to the rest of the grid, the variable and clause gadgets will still interact with the rest of the grid by the crossings that are allowed between their edges.
Because of these allowed crossings, the variable part of graph $G$ will be forced to zigzag horizontally back and forth across the grid; every horizontal stretch of this part will correspond to a single variable from the 3-satisfiability instance. Again, because of its interactions with the grid, the clause part of graph $G$ will be forced to zigzag vertically up and down across the grid; every vertical stretch of this part will correspond to a single clause from the 3-satisfiability instance.

\begin{figure}[t]
\centering\includegraphics[height=3in]{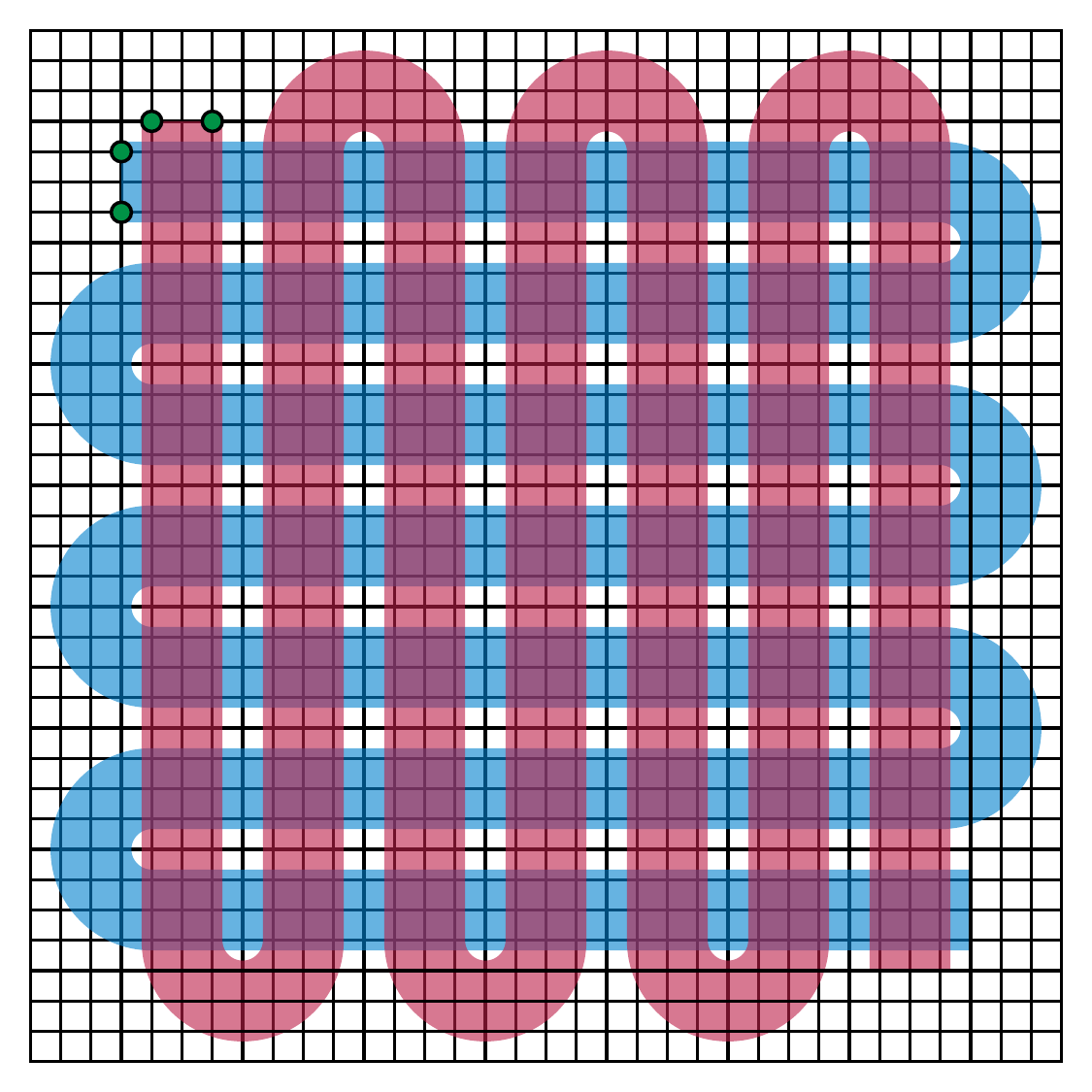}
\caption{Global structure of the graph produced by the \NP-completeness proof for graphs of bounded bandwidth: a grid structure (black) holding in place two paths that zigzag across each other, with each straight segment of these paths either representing a variable (blue) or a clause (red) of a 3-satisfiability instance. The small green circles near the top left indicate the points at which the variable and clause gadgets are attached to the grid; the remaining shape of the variable and clause gadgets is controlled by their allowed crossings with the grid.}
\label{fig:npc-overview}
\end{figure}

Within the variable and clause gadgets of this structure, we will incorporate smaller gadgets that allow each variable or clause gadget to twist relative to itself, so that it has two different possible orientations (twisted or untwisted) as it passes across the grid. In particular, each individual variable clause (a single horizontal stretch of the zigzag pattern formed by the variable part of the graph) will have twist gadgets at either of its two ends, allowing it to take either of these two possible orientations freely. We will use those orientations to encode the truth value of the variable. Each clause gadget will be forced (by its interaction with the grid at either end of its vertical stretch) to have at least one twist, so that the top and the bottom ends of the clause gadget have opposite orientations to each other. We will place twist gadgets within each clause gadget that allow it to have such a twist only when one of the variables has a truth value that satisfies it.

In short, then, we have a grid component of the graph which exists in order to guide the layout of the other parts and make them cross each other in the correct locations. We have a variable gadget for each variable of the 3-sat instance that may take on one of two different orientations according to the 1-planar embedding of the twist gadgets at either of its ends. And, we have a clause gadget for each clause of the 3-sat instance that must twist at least once, and can only twist at the point where it crosses a variable gadget of a variable that belongs to the clause and that has an orientation corresponding to a truth assignment to that variable that would satisfy that clause. As we will describe, it is possible to find a graph of bounded bandwidth that contains gadgets of all these types, and that allows no other 1-planar embeddings than the ones intended to exist as part of the construction. As a result, the graph constructed in this way will have a 1-planar embedding if and only if the given 3-satisfiability instance is satisfiable. This reduction, when complete, will prove the \NP-completeness of 1-planarity for graphs of bounded bandwidth.

\subsection{Crossing control}
\label{sec:npc-cc}

Our \NP-completeness proof involves crossings between several different types of gadgets, and it will be helpful to have some fine-level control of how these crossings may occur. To do so, we introduce a variant of 1-planarity, which we call \emph{colored 1-planarity}, in which the input instance is augmented with edge colors that describe which crossings are allowed. Specifically, in the colored 1-planarity problem, we assume that the edges are labeled from a finite set of colors. One designated color (black, say) is not allowed to participate in any edge crossings; otherwise, an edge may only cross another edge of the same color. The task is to determine whether the given graph has a 1-planar embedding satisfying these color constraints.

\begin{figure}[t]
\centering\includegraphics[height=2in]{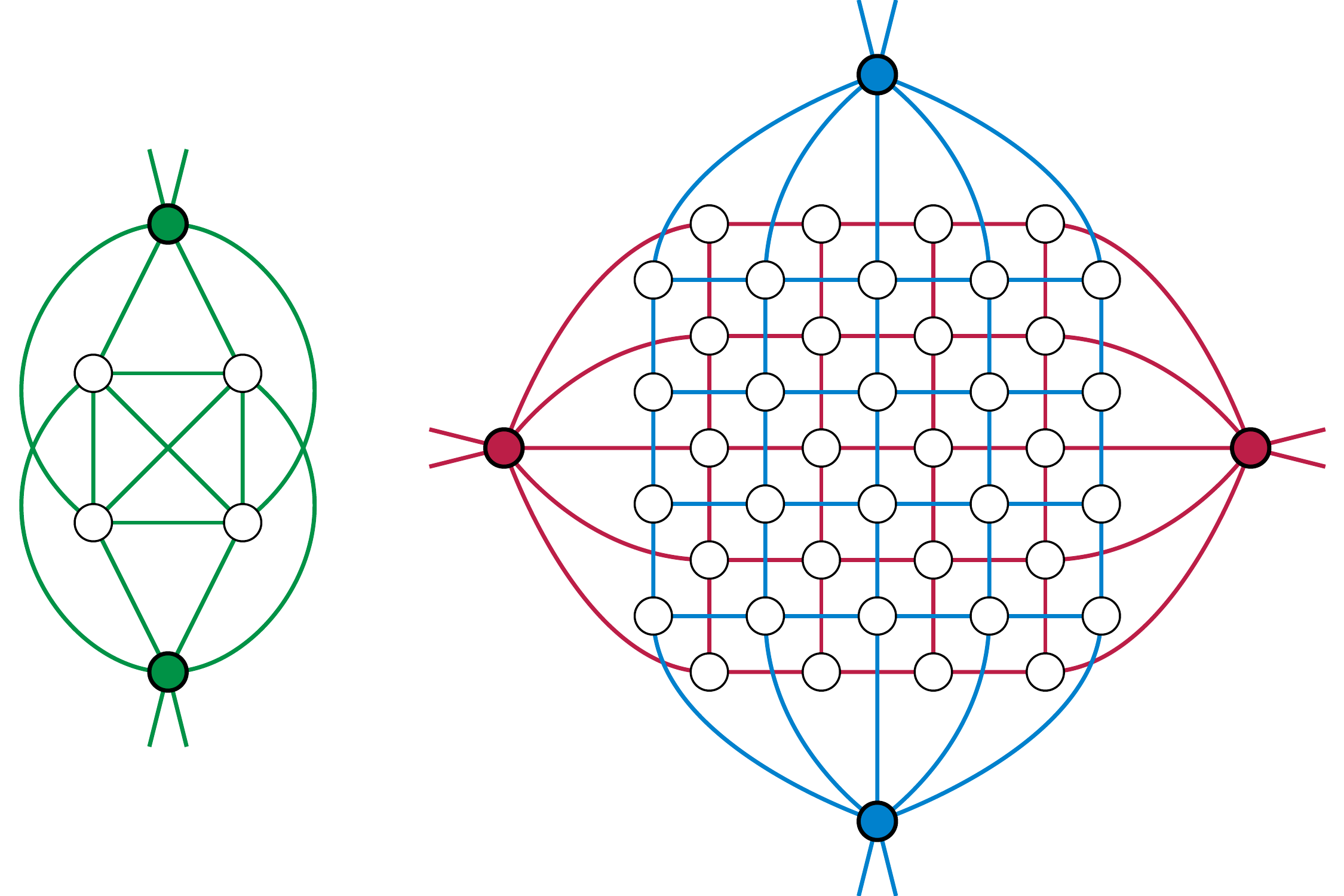}
\caption{Gadgets for reducing colored 1-planarity to uncolored 1-planarity. Left: an uncrossable edge gadget. Right: two crossing grids. In each case the colored vertices indicate the endpoints of the original edge while the uncolored vertices are part of the gadget added to replace it.}
\label{fig:crossing-control}
\end{figure}

\begin{lemma}
\label{lem:crossing-control}
An instance of colored 1-planarity may be reduced to an instance of 1-planarity without colors, preserving the existence or nonexistence of a valid 1-planar drawing, in such a way that the bandwidth of the uncolored instance is $O(1)$ times the bandwidth of the colored instance.
\end{lemma}

\begin{proof}
We replace each edge of the uncrossable color by a gadget whose unique 1-planar embedding does not allow it to be crossed by any other part of a drawing (Figure~\ref{fig:crossing-control}, left). For each other color, we choose an integer $i$ and replace each edge of that color by a grid with $i$ rows and $i+1$ columns, with the two endpoints of the edge connected to the grid points in the two extreme columns of the grid (for instance the red grid in Figure~\ref{fig:crossing-control}, right). Two grids of the same size may cross each other, as shown in the Figure, but it is not possible for grids of different sizes to cross.
\end{proof}

\subsection{Double spiral grid}
\label{sec:npc-grid}

\begin{figure}[t]
\centering\includegraphics[height=3in]{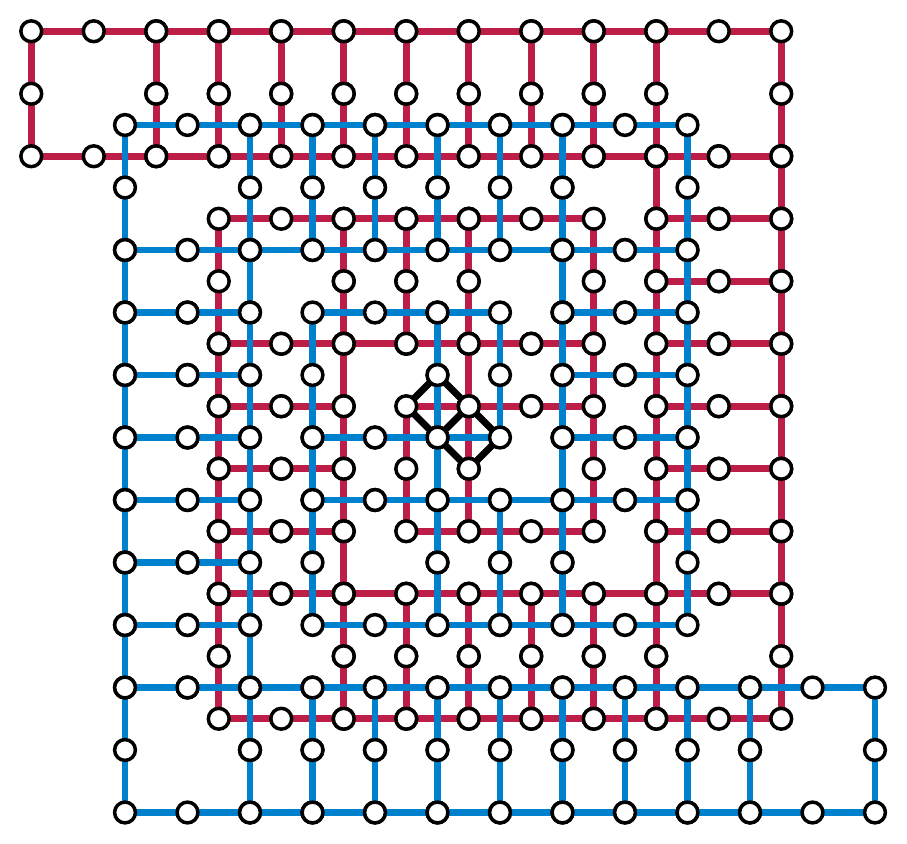}
\caption{Double spiral grid structure. Colors are used to show the two arms or the spiral, but they are unrelated to color 1-planarity.}
\label{figure:low-bandwidth-2x-spiral}
\end{figure}

To form a grid-like substrate on which to build our reduction from 3-satisfiability, we use the double spiral pattern shown in Figure~\ref{figure:low-bandwidth-2x-spiral}, in which two arms spiral from the center of the pattern. We may color the edges of this pattern with a constant number of colors (with color $0$, denoting uncrossable edges, used for edges that do not cross any others) and apply Lemma~\ref{lem:crossing-control} in order to find a graph with the same spiral structure as the figure in which only the crossings shown by the figure are allowed. With this edge coloring, the red and blue spiral arms of the figure (both of which have bounded bandwidth) are forced to continue crossing each other as they spiral around the center, with the leading edge of one arm crossing through the trailing edge of the other arm.

The two arms together generate a grid-like pattern of unbounded size. It is not itself a grid graph (which would have unbounded bandwidth) but its planarization (the planar graph formed by replacing each crossing by a vertex) contains a subdivision of a grid whose size is proportional to the number of windings of the spiral. By additional subdivisions of edges into paths and by using colored 1-planarity, we may allow edge crossings between edges in the double spiral that do not participate in this grid subdivision and edges from other components of the reduction, without allowing unplanned crossings between pairs of edges that both belong to the double spiral. In this way, we may make the parts of the spiral that do not participate in this grid be transparent to the remaining components of the reduction, making the spiral pattern function as a perfect grid for the purposes of describing its interactions with these other components.

In the full reduction, we will also replace some of the edges of this grid subdivision by paths in a colored 1-planarity instance, using a color for each path edge that does not appear anywhere else within the nearby faces of the pattern, in order to allow this grid structure to be crossed in a controlled way by the variable and clause gadgets of the reduction. By using Lemma~\ref{lem:crossing-control} this replacement can be done in a way that does not affect the set of valid 1-planar embeddings of the grid itself.

In order to control the bandwidth of the graph formed by combining this grid gadget with the variable and clause gadgets of the reduction, we require that the attachment points where the variable and clause gadgets attach to the grid gadget (the green circles of Figure~\ref{fig:npc-overview}) lie near each other within the same spiral arm of the grid gadget. This constraint will not cause any theoretical difficulties in the construction.

\subsection{Variable and crossover gadgets}

\begin{figure}[t]
\centering\includegraphics[height=1.5in]{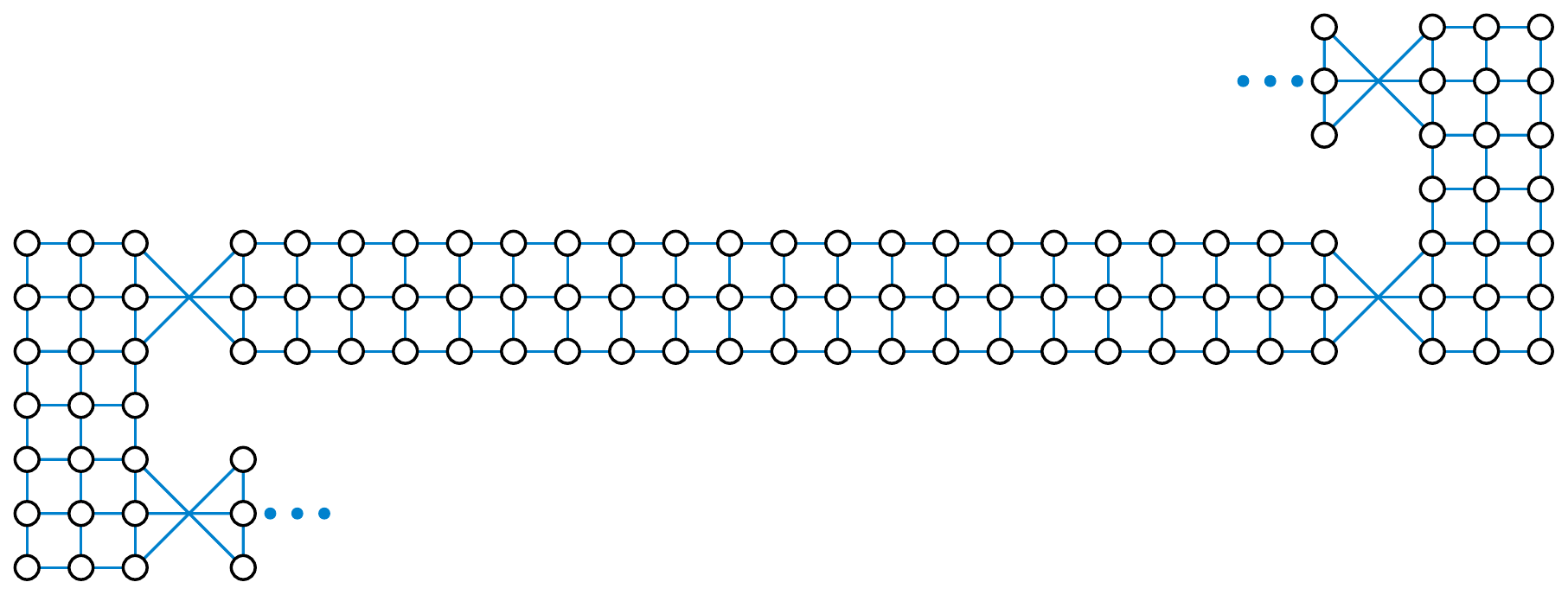}
\caption{Schematic view of the gadget for a single 3SAT variable, consisting of a rigid $3\times i$ grid with crossover gadgets at either end.}
\label{fig:var-gadget-schematic}
\end{figure}

As discussed in Section~\ref{sec:npc-ov}, we will represent the variables of a 3SAT instance by a gadget that is forced (by its set of allowed crossings with the grid gadget) to zigzag horizontally back and forth across the grid of Section~\ref{sec:npc-grid}. Each horizontal stretch of this zigzagging pattern will be formed by a gadget for a single variable, in the form of a $3\times i$ grid of vertices (for some value~$i$ chosen sufficiently large to allow this variable gadget to be crossed by each of the clause gadgets). At either end of this grid, we will have crossover gadgets, shown schematically in Figure~\ref{fig:var-gadget-schematic}, that allow the grid to have one of two possible orientations: either of its two rows of squares may form the top row of squares in the drawing, with the other row of squares below it. These two orientations will correspond to the two truth values of the variable.

\begin{figure}[t]
\centering\includegraphics[width=4.5in]{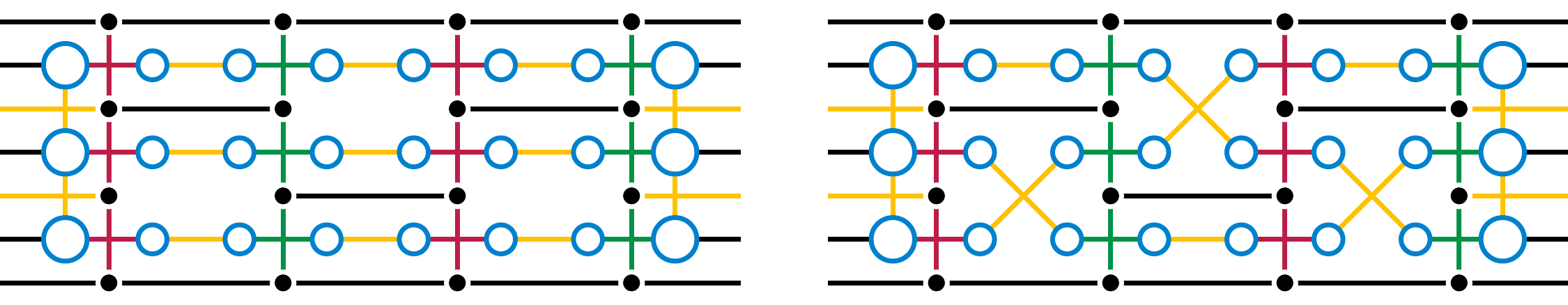}
\caption{The crossover gadget for the two ends of a variable gadget, in its uncrossed (left) and crossed (right) states.}
\label{fig:crossover-gadget}
\end{figure}

Within the variable gadget, there are three parallel and disjoint paths of length $i$, extending horizontally across the gadget. To allow the gadget to have its two different orientations, these three paths must be allowed to cross each other within the crossover gadget. However, it would not work to allow the crossings between these three paths to lie close enough to each other that they belong to a single face of the underlying grid (as might be suggested by the drawing of Figure~\ref{fig:var-gadget-schematic}, which shows all three paths crossing at a single point). The reason is that, if all three paths of the gadget were allowed to pass through a single face of the grid, it would then be possible to continue drawing the rest of the variable gadget, and the sequence of variable gadgets connected to it, all within this same face, thwarting our intended zigzag pattern. In order to control the global shape formed by the variable gadgets as they extend across the grid, we must make sure that only two of the three paths ever lie in a single face of the grid.

A crossover gadget that allows the variable gadget to take either of its two orientations, but does not allow it to escape into a single face of the grid, is shown in Figure~\ref{fig:crossover-gadget}. The colored edges of the figure are used to describe pairs of edges that may cross, in a colored 1-planarity problem, as described in Section~\ref{sec:npc-cc}. Note that the gadget consists of two disconnected subgraph: the blue vertices within a variable gadget and the black vertices within the grid. As long as the three leftmost vertices of the variable gadget are constrained (by earlier parts of the construction) to lie in the three grid faces that they are shown in, the rest of the crossover gadget must extend rightwards across the grid as shown, in order to reach another set of faces from which the three right vertices can be connected to each other. The crossings in the middle of the gadget would allow the three right vertices to be arranged in any of the six possible permutations of the three left vertices, but only the original permutation and its reversal allow a consistent placement of the yellow edges at the right of the figure.

\subsection{Clause and gated crossover gadgets}

\begin{figure}[t]
\centering\includegraphics[width=3in]{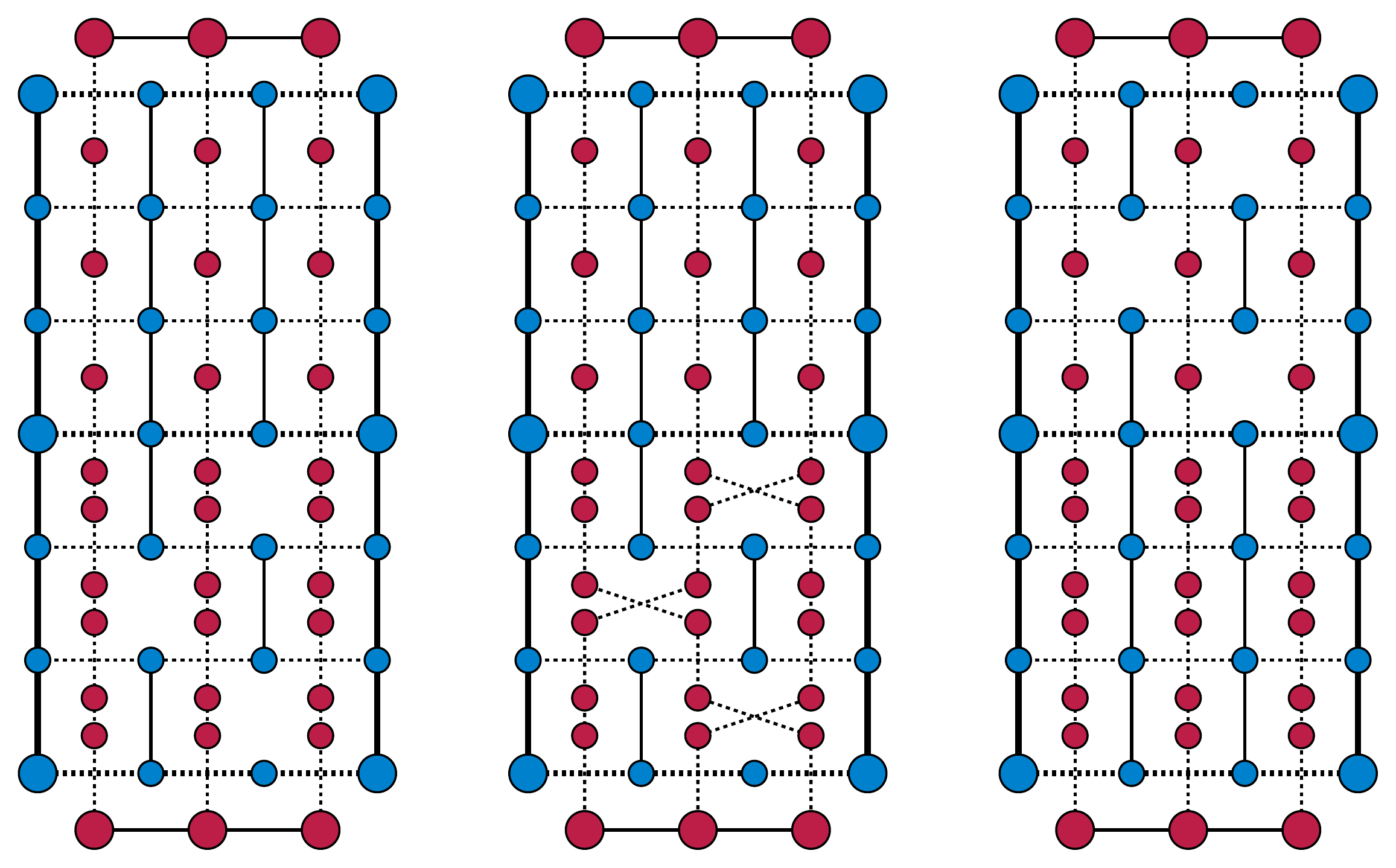}
\caption{Three states of a gated crossover gadget. Left: the variable gadget's orientation corresponds to a truth value that satisfies the clause, but the clause gadget is uncrossed. Center: the variable satisfies the clause, and the clause gadget is crossed. Right: the variable does not satisfy the clause; no crossing is possible. The edges of the figure are colored to form a colored 1-planarity problem that prevents crossings other than the ones shown between crossing pairs of dashed edges. The underlying grid and its crossings within the gadget are not depicted.}
\label{fig:gated-crossover}
\end{figure}

As discussed in Section~\ref{sec:npc-ov}, we will represent the clauses of a 3SAT instance by a gadget that is forced (by its set of allowed crossings with the grid gadget) to zigzag vertically up and down across the grid of Section~\ref{sec:npc-grid}. Each vertical stretch of this zigzagging pattern will be formed by a gadget for a single clause, and (like the variable gadgets) will consist of a $3\times i$ grid of vertices, for some appropriate choice of~$i$, interspersed with \emph{gated crossover} gadgets, one for each term in the clause. The top and bottom connections of this gadget with the adjacent clauses will form more grid graphs that are rigid in their layout with respect to the grid gadget, forcing the clause to twist an odd number of times within its gated crossover gadgets. A twist will only be allowed within a gated crossover gadget when the variable gadget corresponding to one of the terms in the clause has the correct orientation (corresponding to a truth assignment for which that variable satisfies the clause). In this way, it will only be possible to find a 1-planar layout for all of the clause gadgets if the variable gadgets are oriented in a way that corresponds to a satisfying assignment for the whole 3SAT instance.

To form a gated crossover gadget, we replace two vertically-adjacent squares of the grid of a variable gadget, as shown in Figure~\ref{fig:gated-crossover}. One of the two squares (the one that is the bottommost of the two squares for truth assignments that cause the variable to satisfy the clause) is replaced by the grid part of a crossover gadget, rotated by $90^\circ$ from the ones in Figure~\ref{fig:crossover-gadget}, while the other square is replaced by a modified crossover gadget with extra edges that prevents its crossed state from being a valid layout. The clause gadget forms three paths that pass through both of these crossover gadgets, with their edges colored to form a colored 1-planarity instance in such a way that the parts of the paths within the top square are not allowed to cross each other (they can only make the required crossings with the gadget) while the paths may cross each other within the bottom square. Thus, when the gadget is aligned so that the bottom square is the one containing the unmodified crossover gadget, the three clause gadget paths may cross each other or not, but when it is aligned so that the top square is the one containing the unmodified crossover gadget, no crossing is possible.

The parts of the gated crossover gadget that come from the variable and clause gadgets, shown in the figure, are overlaid by additional parts coming from the grid gadget, which force the left and right sides of the variable gadget to spread across multiple grid faces, and also force the top and bottom sides of the clause gadget to spread across multiple grid faces, preventing additional unwanted 1-planar layouts where part of a variable or clause gadget and all the rest of the gadgets connected in sequence to it lie within a single face; we omit the details, as they complicate the gated crossover gadget without significantly affecting the mechanism by which it works.

\subsection{\NP-completeness}

\begin{theorem}
It is \NP-complete to test 1-planarity for graphs of bounded bandwidth.
\end{theorem}

\begin{proof}
We use a reduction from 3SAT as described above, by forming a graph from the grid gadgets, variable gadgets, and clause gadgets described in the previous subsections. The subgraph formed by the sequence of variable gadgets, and the subgraph formed by the sequence of clause gadgets, are connected at an appropriate place (near one corner of the grid) to the subgraph formed by the grid gadget; otherwise these three graphs are disconnected from each other, and interact only through their crossings. The variable and clause gadgets are forced by their crossings with the grid gadget to take zigzagging paths across the grid, so that each variable gadget crosses each clause gadget in a controlled area within the grid.

Each variable gadget may have one of two orientations (given by the crossover gadgets at either of its ends), one of which is associated with a true value of the variable and the other of which is associated with a false value. When a variable gadget crosses a clause gadget that has a term containing that variable, the gated crossover gadget associated with that crossing (consisting of subconfigurations within both the variable and clause gadget) allows the clause gadget to twist at that point, if and only if the variable's truth assignment would satisfy that clause. For variables with the wrong truth assignment for the given clause, or that do not participate in the clause, no such twist is possible. The parts of the graph that connect consecutive clause gadgets are arranged in such a way that they can be embedded in a 1-planar way only if the two ends of the clause gadget have the correct orientations, with a single twist relative to each other.

If the input 3SAT instance has a satisfying assignment, it can be used to choose twists for the crossover and gated crossover gadgets of this reduction in such a way that the entire graph is embedded in a 1-planar way. On the other hand, any embedding of the graph must come from a satisfying truth assignment in this way. The translation described by this reduction can be performed by a polynomial time algorithm, as it involves only putting the correct gadgets together in the correct sequence. Therefore, this translation is a valid many-one reduction from 3SAT to 1-planarity. The graph that results from the reduction consists of a bounded number of bounded-bandwidth pieces (the grid gadget, sequence of variable gadgets, and sequence of clause gadgets), whose bandwidth is expanded by a constant factor due to the reduction from colored 1-planarity to uncolored 1-planarity used to control the pairs of edges that are allowed to cross. Moreover the pieces can be joint keeping the bandwidth bounded: take an arrangement placing the variable zigzag  first, with the green vertices of attachment to the grid at the end, followed with an arrangement of the grid with the green vertices of attachment to the variable zigzag at the beginning and the green vertices of attachment to the clause zigzag at the end, and finish with an arrangement of the clause zigzag.
Therefore, the whole graph has bounded bandwidth overall. 
\end{proof}

\fi 

\end{document}